\documentclass[11pt]{article}

\usepackage[hyphens]{url}
\usepackage[pagebackref,colorlinks]{hyperref}
\usepackage[hyphenbreaks]{breakurl}
\usepackage{amsmath} 
\usepackage{amsthm} 
\usepackage{amssymb}	
\usepackage{graphicx} 
\usepackage{multicol} 
\usepackage{multirow}
\usepackage{color}
\usepackage[dvips,letterpaper,margin=1in,bottom=1in]{geometry}
\usepackage[capitalize,noabbrev]{cleveref}

\usepackage{bm}

\usepackage{ytableau}

\usepackage[utf8]{inputenc}
\usepackage[english]{babel}

\usepackage[T1]{fontenc}
\AtBeginDocument{%
  \DeclareFontShape{T1}{cmr}{m}{scit}{<->ssub*cmr/m/sc}{}%
}

\usepackage{diagbox}
\usepackage{mathtools}

\newtheorem{theorem}{Theorem}[section]

\newtheorem{lemma}[theorem]{Lemma}
\newtheorem{corollary}[theorem]{Corollary}

\newtheorem{definition}[theorem]{Definition}

\newtheorem{remark}{Remark}[section]
\newtheorem{problem}{Problem}

\newcommand{\braket}[2]{\langle #1 | #2 \rangle}

\DeclarePairedDelimiter\rbra{\lparen}{\rparen}
\DeclarePairedDelimiter\sbra{\lbrack}{\rbrack}
\DeclarePairedDelimiter\cbra{\{}{\}}
\DeclarePairedDelimiter\abs{\lvert}{\rvert}
\DeclarePairedDelimiter\Abs{\lVert}{\rVert}
\DeclarePairedDelimiter\ceil{\lceil}{\rceil}

\DeclarePairedDelimiter\ket{\lvert}{\rangle}
\DeclarePairedDelimiter\bra{\langle}{\rvert}

\DeclareMathOperator*{\E}{\mathbb{E}}
\DeclareMathOperator*{\Var}{\mathbb{V}\!\!\;ar}

\newcommand{\tr} {\operatorname{tr}}

\newcommand{\supp} {\operatorname{supp}}
\newcommand{\spanspace} {\operatorname{span}}

\newcommand{\ketbra}[2]{\ensuremath{\ket{#1}\!\bra{#2}}}

\usepackage{latexsym}
\usepackage{CJK}

\usepackage{enumerate}

\usepackage{algorithm}
\usepackage{algpseudocode}

\usepackage{stmaryrd}
\usepackage{tabularx}
\usepackage{booktabs}
\usepackage{adjustbox}

\newcommand{\footremember}[2]{%
    \footnote{#2}
    \newcounter{#1}
    \setcounter{#1}{\value{footnote}}%
}

\usepackage{tikz}
\usetikzlibrary{quantikz2}

\begin{document}

\title{Estimating Fidelity to a Reference Quantum State}
\author{Qisheng Wang \footremember{2}{Qisheng Wang is with the School of Computer Science, Shanghai Jiao Tong University (e-mail: \url{QishengWang1994@gmail.com}).}}
\date{}

\maketitle

\begin{abstract}
    We consider the problem of estimating the fidelity of an unknown quantum state to a known reference state to within additive error $\varepsilon$. 
    We show that the sample complexity is $O(r^2/\varepsilon^2)$ with \textit{optimal} $\varepsilon$-dependence when the reference state is of rank $r$, improving the previous best $O(r^2\log^2(1/\varepsilon)/\varepsilon^4)$ due to \hyperlink{cite.UNWT25}{Utsumi, Nakata, Wang, and Takagi (QIP 2026)}. 
    We also provide a lower bound of $\Omega(r/\varepsilon^2)$, improving the previous best $\Omega(r/\varepsilon+1/\varepsilon^2)$, with implications to quantum query complexity. 
    Moreover, we further consider the case where the unknown state is of rank at most $r$ while the reference state can be arbitrary, for which the sample complexity is shown to be $O(r^2/\varepsilon^4)$. 
    As an application, we present an approach to tolerant quantum state certification, generalizing the exact certification studied in \hyperlink{cite.BOW19}{B\u{a}descu, O'Donnell, and Wright (STOC 2019)}. 
\end{abstract}

\newpage
\tableofcontents
\newpage

\section{Introduction}

The (Uhlmann) fidelity \cite{Uhl76,Joz94} is an information-theoretic quantity that measures the closeness between quantum states, which generalizes the Bhattacharyya coefficient between probability distributions \cite{Bat46}. 
The fidelity between two quantum states $\rho$ and $\sigma$ is defined by (see \cite[Equation (9.53)]{NC10})
\[
\mathrm{F}\rbra{\rho, \sigma} = \tr\rbra*{\sqrt{\sqrt{\sigma}\rho\sqrt{\sigma}}}. 
\]
This quantity has been widely used in quantum physics and is now a standard in quantum computation and quantum information (cf.\ \cite{NC10}). 
Many fundamental tasks such as quantum state tomography \cite{HHJ+17,OW16,SSW25,PSW25,PSTW25} and quantum state certification \cite{BOW19} are investigated with respect to the fidelity. 

Fidelity estimation therefore becomes an important problem in quantum computing. 
One of the earliest efficient approaches is the pure-state fidelity estimation based on the SWAP test \cite{BCWdW01}. 
However, it has been known that fidelity estimation is $\mathsf{QSZK}$-hard in general \cite{Wat02,Wat09}, which means that one cannot hope for a generally efficient approach to fidelity estimation unless $\mathsf{BQP} = \mathsf{QSZK}$.
Therefore, the consequent work in the literature focused on finding efficient approaches in certain useful cases.
In particular, it is known that low-rank fidelity estimation is $\mathsf{BQP}$-complete due to the results in \cite{WZC+23,RASW23}. 
From the perspective of computational complexity, a fundamental question is how many samples of quantum states (i.e., the sample complexity) are needed to estimate the fidelity between them.\footnote{There is another complexity measure, called quantum query complexity, that has been studied in the literature for fidelity estimation. We will review this line of research as related work in \cref{sec:related-work}.} 
There has been a series of work on the sample complexity of fidelity estimation. 

\begin{itemize}
    \item For pure states, the fidelity can be estimated to within additive error $\varepsilon$ with sample complexity $O\rbra{1/\varepsilon^4}$ through the SWAP test \cite{BCWdW01}, which was recently improved to $O\rbra{1/\varepsilon^2}$ in \cite{WZ24} matching the lower bound given in \cite{ALL22}. 
    There are other approaches with restricted measurements, including the direct fidelity estimation method \cite{FL11} with only Pauli measurements, the entanglement witness method \cite{GT09} with few measurements, and the distributed inner product estimation \cite{ALL22} with local operations and classical communication (LOCC) and its extensions \cite{AS25,GHYZ24,ZWY+25}. 

    \item For low-rank states, when the quantum states are of rank $r$, the fidelity can be estimated to within additive error $\varepsilon$ with sample complexity $\widetilde{O}\rbra{r^{5.5}/\varepsilon^{12}}$ \cite{GP22}, which was recently improved to $O(r^2\log^2(1/\varepsilon)/\varepsilon^4)$ in \cite{UNWT25}, whereas the prior best lower bound is $\Omega\rbra{r/\varepsilon+1/\varepsilon^2}$ due to \cite{OW21,ALL22}. 
    When the reciprocal of the smallest eigenvalue of the quantum states is $\kappa$, the fidelity can be estimated with sample complexity $\widetilde{O}\rbra{\kappa^9/\varepsilon^3}$ \cite{LWWZ25}, which was later improved to $O\rbra{\kappa^2\log^2\rbra{1/\varepsilon}/\varepsilon^2}$, almost matching the lower bound $\Omega\rbra{1/\varepsilon}$ in \cite{LWWZ25} for constant $\kappa$. 
\end{itemize}

\subsection{Main results}

In this paper, we consider the fidelity estimation where one quantum state $\sigma$ is known, called the ``reference'' state, while the other quantum state $\rho$ is unknown. 
The fidelity to a reference state turns out to be a common indicator for evaluating quantum state tomography \cite{FR16,WSR18} and entanglement detection \cite{EG23}. 
We formalize the problem as follows. 

\begin{problem} [Estimating fidelity to a reference] \label{prob:main}
    Given a known reference quantum state $\sigma$ (with its classical description) and an unknown quantum state $\rho$ (with its identical copies), the task is to estimate the fidelity $\mathrm{F}\rbra{\rho, \sigma}$ to within additive error $\varepsilon$ using as few samples of $\rho$ as possible.
\end{problem}

We present an approach to \cref{prob:main} that estimates the fidelity of an unknown quantum state to a reference state, with sample complexity depending on the rank of the reference state. 

\begin{theorem}[Fidelity estimation with a low-rank reference state, \cref{thm-11242222} and \cref{thm-240417} restated] \label{thm:fidelity-main}
    Given a known reference quantum state $\sigma$ of rank $r$ and an unknown quantum state $\rho$, the fidelity $\mathrm{F}\rbra{\rho, \sigma}$ can be estimated to within additive error $\varepsilon$ using $O\rbra{r^2/\varepsilon^2}$ samples of $\rho$. 
    Moreover, $\Omega\rbra{r/\varepsilon^2}$ samples of $\rho$ are necessary. 
\end{theorem}

Compared to \cref{thm:fidelity-main}, the prior best approach to fidelity estimation \cite{UNWT25} only gives a sample complexity upper bound of $O\rbra{r^2\log^2\rbra{1/\varepsilon}/\varepsilon^4}$ and the prior best sample complexity lower bound is $\Omega\rbra{r/\varepsilon+1/\varepsilon^2}$ due to \cite{OW21,ALL22} (as noted in \cite{UNWT25}). 
It is worth noting that \cref{thm:fidelity-main} gives the \textit{first} dimension-independent approach to mixed-state fidelity estimation  with optimal $\varepsilon$-dependence. 
In addition, the $r$-dependence of \cref{thm:fidelity-main} also matches the sample lower bound $\Omega\rbra{r^2}$ for fidelity estimation conjectured in \cite{CWZ25}. 

\paragraph{Extensions.}
Furthermore, we extend \cref{thm:fidelity-main} to the case where the unknown quantum state $\rho$ is guaranteed to be of rank $\leq r$ but the known reference quantum state $\sigma$ can be arbitrary. 

\begin{theorem}[Fidelity estimation for low-rank unknown states, \cref{thm-1131746} restated] \label{thm:fidelity-ext}
    Given a known reference quantum state $\sigma$ and an unknown quantum state $\rho$ that is guaranteed to be of rank $\leq r$, the fidelity $\mathrm{F}\rbra{\rho, \sigma}$ can be estimated to within additive error $\varepsilon$ using $O\rbra{r^2/\varepsilon^4}$ samples of $\rho$. 
\end{theorem}

Compared to the prior best approach to fidelity estimation \cite{UNWT25} with sample complexity $O\rbra{r^2\log^2\rbra{1/\varepsilon}/\varepsilon^4}$, \cref{thm:fidelity-ext} removes the polylogarithmic factor in $\varepsilon$. 
Combining both \cref{thm:fidelity-main,thm:fidelity-ext}, we obtain the \textit{current best} sample complexity upper bound for low-rank fidelity estimation when one of the quantum states is known.

\vspace{10pt}

In \cref{tab:cmp}, we compare our results with previous approaches. 
Further applications and extensions of \cref{thm:fidelity-main} will be discussed later in \cref{sec:app,sec:related-work}, respectively. 

\begin{table}[!htp]
    \centering
    \caption{Sample complexity of mixed-state fidelity estimation.}
    \label{tab:cmp}
    \adjustbox{max width=\textwidth}{
    \begin{tabular}{cccc}
        \toprule
        Problem & Sample Complexity & Constraints & Reference \\
        \midrule
        \multirow{4}{*}{Fidelity Estimation} & $\widetilde{O}\rbra{r^{5.5}/\varepsilon^{12}}$ & \multirow{2}{*}{Either $\rho$ or $\sigma$ is of rank $r$.} & \cite{GP22} \\
        & $O\rbra{r^{2}\log^2\rbra{1/\varepsilon}/\varepsilon^{4}}$ & & \cite{UNWT25} \\
        \cmidrule{2-4}
        & $O\rbra{r^{2}/\varepsilon^{2}}$, $\Omega\rbra{r/\varepsilon^2}$ & $\sigma$ is known and of rank $r$. & \cref{thm:fidelity-main} \\
        \cmidrule{2-4}
        & $O\rbra{r^2/\varepsilon^4}$ & $\sigma$ is known and $\rho$ is of rank $r$. & \cref{thm:fidelity-ext} \\
        \bottomrule
    \end{tabular}
    }
\end{table}

\subsection{Applications to tolerant quantum state certification} \label{sec:app}

The problem of estimating the fidelity to a reference state (\cref{prob:main}) directly relates to the tolerant version of quantum state certification \cite{BOW19}.
As shown in \cref{fig:qsc}, the exact quantum state certification with parameter $\varepsilon$ is a special case of tolerant quantum state certification with the parameters $\varepsilon_1 = 0$ and $\varepsilon_2 = \varepsilon$. 

\begin{figure}[!htp]
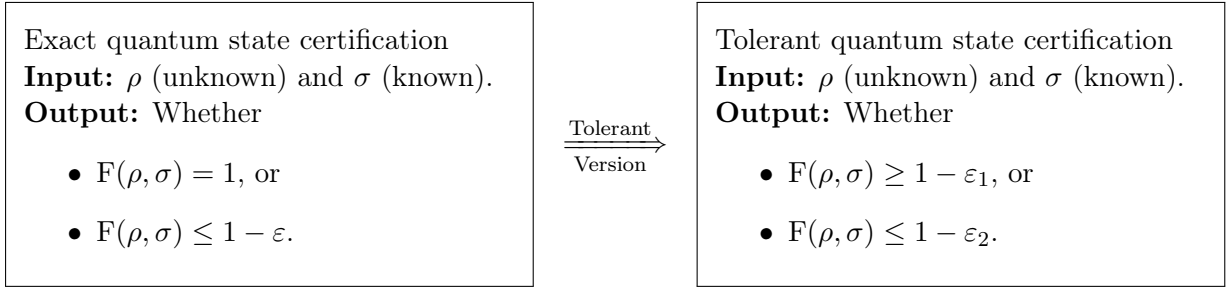

    \centering
\vspace{5pt}
    \noindent\fbox{
    \parbox{6.5cm}{
    \vspace{7pt}
    
    Exact quantum state certification

    \textbf{Input:} $\rho$ (unknown) and $\sigma$ (known). 

    \textbf{Output:} Whether 
    \begin{itemize}
        \item $\mathrm{F}\rbra{\rho, \sigma} = 1$, or
        \item $\mathrm{F}\rbra{\rho, \sigma} \leq 1 - \varepsilon$. 
    \end{itemize}
    
    \vspace{2pt}
    }
    }
    ~ $\xRightarrow[\textup{Version}]{\textup{Tolerant}}$ ~
    \fbox{
    \parbox{6.5cm}{
    \vspace{7pt}
    
    Tolerant quantum state certification

    \textbf{Input:} $\rho$ (unknown) and $\sigma$ (known). 

    \textbf{Output:} Whether 
    \begin{itemize}
        \item $\mathrm{F}\rbra{\rho, \sigma} \geq 1 - \varepsilon_1$, or
        \item $\mathrm{F}\rbra{\rho, \sigma} \leq 1 - \varepsilon_2$. 
    \end{itemize}
    
    \vspace{2pt}
    }
    }
    \vspace{5pt}
    \caption{Exact and tolerant quantum state certifications.}
    \label{fig:qsc}
\end{figure}

As an application, we obtain an approach to tolerant quantum state certification. 
\begin{corollary}[Tolerant quantum state certification] \label{corollary:tolerant-qsc}
    Given a known reference quantum state $\sigma$ of rank $r$ and an unknown quantum state $\rho$, whether $\mathrm{F}\rbra{\rho, \sigma} \geq 1-\varepsilon_1$ or $\mathrm{F}\rbra{\rho, \sigma} \leq 1-\varepsilon_2$ can be determined using $O\rbra{r^2/\rbra{\varepsilon_2-\varepsilon_1}^2}$ samples of $\rho$. 
\end{corollary}

In \cite{BOW19}, they showed that the sample complexity of (exact) quantum state certification (with respect to fidelity) is $\Theta\rbra{d/\varepsilon}$ for $d$-dimensional quantum states, in which case $\varepsilon_1 = 0$ and $\varepsilon_2 = \varepsilon$. 
In comparison, \cref{corollary:tolerant-qsc} is a tolerant generalization of the results in \cite{BOW19} and further considers the $r$-dependence for the case where the rank $r$ of the reference state is much smaller than the dimension $d$. 

\subsection{Implications to quantum query complexity} \label{sec:related-work}

In addition to sample complexity, the quantum query complexity of fidelity estimation has also been extensively studied in the literature.
In this case, query access to the state-preparation circuits of (the purifications of) the quantum states is given and the query complexity measures how many queries to the circuits are used. 
This is called the ``purified quantum query access'' model \cite{GL20}, which is now the standard quantum query model for testing quantum states. 
The quantum query complexity of fidelity estimation has also been investigated in the literature. 
\begin{itemize}
    \item For pure states, the fidelity can be estimated to within additive error $\varepsilon$ with query complexity $O\rbra{1/\varepsilon^2}$ by combining the SWAP test \cite{BCWdW01} and the quantum amplitude estimation \cite{BHMT02}, which was recently improved to the optimal $\Theta\rbra{1/\varepsilon}$ in \cite{Wan24,FW25}. 
    \item For low-rank states, the first polynomial-time quantum algorithm was given in \cite{WZC+23} with query complexity $\widetilde{O}\rbra{r^{12.5}/\varepsilon^{13.5}}$, which was later improved to $\widetilde{O}\rbra{r^{6.5}/\varepsilon^{7.5}}$ in \cite{WGL+24}, to $\widetilde{O}\rbra{r^{2.5}/\varepsilon^5}$ in \cite{GP22}, and recently to $\widetilde{O}\rbra{r/\varepsilon^2}$ in \cite{UNWT25}. 
    When the reciprocal of the smallest eigenvalue of the quantum states is $\kappa$, the query complexity of fidelity estimation was shown to be $\widetilde{O}\rbra{\kappa^4/\varepsilon}$ \cite{LWWZ25}, which was later improved to $\widetilde{O}\rbra{\kappa/\varepsilon}$ in \cite{UNWT25}, with near-optimal $\varepsilon$-dependence. 
\end{itemize}

On the other hand, the lower bound on the query complexity of fidelity estimation is trivially $\Omega\rbra{r^{1/3}+1/\varepsilon}$ (as mentioned in \cite{UNWT25}) due to the lower bounds for quantum counting \cite{BBC+01,NW99} and the uniformity testing of probability distributions \cite{CFMdW10}, which was recently improved to $\Omega\rbra{\sqrt{r/\varepsilon}+1/\varepsilon}$ in \cite{CWZ25}.
With the sample complexity lower bound in \cref{thm:fidelity-main}, we can further obtain a slightly improved quantum query complexity of fidelity estimation. 

\begin{corollary}
    The quantum query complexity of fidelity estimation is $\Omega\rbra{\sqrt{r}/\varepsilon}$. 
\end{corollary}
\begin{proof}
    Using the same arguments in \cite{CWZ25} based on the quantum sample-to-query lifting \cite{WZ25a,WZ25b,TWZ25}, it is known that the quantum query complexity of a property testing problem is at least the square root of the sample complexity of the same problem. 
    Therefore, the sample complexity lower bound $\Omega\rbra{r/\varepsilon^2}$ for fidelity estimation in \cref{thm:fidelity-main} immediately gives a quantum query complexity lower bound of $\Omega\rbra{\sqrt{r}/\varepsilon}$ for fidelity estimation. 
\end{proof}

\subsection{Techniques}
\label{sec:technique-overview}

\subsubsection{Reduce fidelity estimation to a low-dimensional problem}
A key structural observation is that, when the reference state $\sigma$ is known, the fidelity can be computed ``locally'' on an appropriate low-dimensional subspace. We will explicitly construct this subspace later; for now, let $\Pi$ denote the orthogonal projector onto it. 
We can show that the contribution of $\rho$ outside $\mathrm{supp}(\Pi)$ is negligible, and one can work with the compressed operator $\Pi \rho \Pi$ instead of $\rho$, i.e., 
\[\mathrm{F}(\rho,\sigma)\approx \mathrm{F}(\Pi\rho\Pi,\sigma).\]

Based on this, we use a \emph{two-stage} procedure:
\begin{enumerate}
  \item \textbf{Filtration:} Perform the binary measurement
  $\{\Pi,\, I-\Pi\}$ on copies of $\rho$,
  and estimate the acceptance probability
  $
    p := \mathrm{tr}(\Pi \rho).
  $
  \item \textbf{Local subspace tomography:} Conditioned on acceptance, the
  post-measurement state is
  \[
    \rho_\Pi := \frac{\Pi \rho \Pi}{\mathrm{tr}(\Pi \rho)} = \frac{\Pi \rho \Pi}{p},
  \]
  which lives in the (typically low-dimensional) subspace $\supp(\Pi)$.
  We then perform tomography \emph{only within this subspace} to obtain an estimator
  $\widehat{\rho}_\Pi$ that is close to $\rho_\Pi$ in an operationally meaningful metric. 
\end{enumerate}

However, this approach faces several challenges.  
First, performing standard tomography of $\rho_{\Pi}$ with guarantees in trace distance (cf.\ \cite{OW16}) does not, in general, yield an accurate estimate of the fidelity $\mathrm{F}(\rho,\sigma)$.  
Second, when $p$ is small, it becomes inefficient to obtain a sufficient number of copies of $\rho_{\Pi}$ to carry out high-precision tomography.

To address the first difficulty, we adopt the \emph{Bures distance} as the tomography metric. This choice is motivated by the fact that the Bures distance is directly related to fidelity and satisfies the triangle inequality. Consequently, a tomography guarantee in Bures distance can be translated into a bound on the resulting error in the fidelity estimate.
Specifically, we use the recent tomography result in the Bures distance in \cite{PSW25}, which removes logarithmic factors from the sample complexity of the approach in \cite{HHJ+17}. 

For the second difficulty, we note that the target fidelity depends on $p$, and when $p$ is small, the contribution of any estimation error in $\rho_\Pi$ is correspondingly attenuated. In particular, we exploit the (generalized) homogeneity of fidelity:
$\mathrm{F}(\Pi\rho\Pi,\sigma)=\sqrt{p}\,\mathrm{F}(\rho_\Pi,\sigma)$.
We then define the estimator
\[
\widehat{\mathrm{F}} \coloneqq \sqrt{\widehat{p}}\;\mathrm{F}\!\left(\widehat{\rho}_\Pi,\sigma\right),
\]
where $\widehat{p}$ is the empirical estimator of $p$. With a careful error analysis, we show that using $\widehat{p}$ effectively suppresses the impact of the estimation error in $\rho_\Pi$, ensuring that $\widehat{\mathrm{F}}$ provides an accurate estimate of $\mathrm{F}(\rho,\sigma)$.

\paragraph{Scenario I: Low-rank reference state $\sigma$, arbitrary unknown state $\rho$.}
Assume $\mathrm{rank}(\sigma)= r$ and that $\sigma$ is known classically.
We take $\Pi$ as the orthogonal projector onto $\supp(\sigma)$, so the tomography step only needs to operate in an $r$-dimensional
subspace.
This yields a sample complexity scaling
\[
n = O\!\left(\frac{r^2}{\varepsilon^2}\right),
\]
with no dependence on the ambient dimension $d$.
At a high level, the error budget splits into:
(i) concentration of $\widehat{p}$ around $p$, and
(ii) the tomography error in estimating $\rho_\Pi$ in the Bures distance, 
which is then propagated to the final estimate of the fidelity. 

\paragraph{Scenario II: Low-rank unknown state $\rho$, arbitrary reference state $\sigma$.}
When $\sigma$ is allowed to be full-rank, projecting onto $\mathrm{supp}(\sigma)$ can be high-dimensional.
Since instead $\rho$ is promised to satisfy $\mathrm{rank}(\rho)\le r$, we can construct a
\emph{dimension-reducing truncation} of $\sigma$.
Specifically, let $K := \Theta({r}/{\varepsilon^2})$ and $\Pi$ be the orthogonal projector onto the subspace spanned by the top-$K$ eigenvectors (i.e., the eigenvectors of the top-$K$ largest eigenvalues) of $\sigma$.
Then, we define
\[
\sigma' := \Pi \sigma \Pi.
\]
One can show that, under $\mathrm{rank}(\rho)\le r$, truncating $\sigma$ in this way perturbs
the fidelity by at most $O(\varepsilon)$, i.e.,
\[
\bigl|\mathrm{F}(\rho,\sigma) - \mathrm{F}(\rho,\sigma')\bigr| \le O(\varepsilon).
\]
We then estimate $\mathrm{F}(\rho,\sigma')$ using the same two-stage procedure as above.
Since tomography now runs in a $K$-dimensional subspace, with a careful error analysis, this leads to the sample complexity scaling
\[
n = O\!\left(\frac{r^2}{\varepsilon^4}\right).
\]

\subsubsection{Intuition for lower bounds}

Our lower bound $n = \Omega\rbra{r/\varepsilon^2}$ on the sample complexity of fidelity estimation with optimal $\varepsilon$-dependence holds even when $\sigma$ is known of an arbitrary rank $r$. 
The proof proceeds by reduction from a quantum hypothesis testing problem considered in \cite{CHW07} with sample complexity $\Theta\rbra{r}$, which is to determine whether an unknown quantum state has a uniform spectrum on (i) $r$ or (ii) $2r$ eigenvalues. 
In our case, we set the reference state $\sigma = \frac{1}{2r} \sum_{i=1}^{2r} \ketbra{i}{i}$. 
To determine whether an unknown quantum state $\rho$ is in case (i) or (ii), we use the procedure that generates $\Theta\rbra{S/\varepsilon^2}$ copies of $\rho_\varepsilon = \rbra{1-\varepsilon^2} \ketbra{0}{0} + \varepsilon^2 \rho$ from $S$ copies of $\rho$; this is inspired by \cite{OW21} where they considered the rank testing problem. 
Then, the two cases can be distinguished by considering the value of $\mathrm{F}\rbra{\rho_\varepsilon, \sigma}$, which can be obtained by fidelity estimation, and thus a sample complexity lower bound for fidelity estimation can be established.

\subsection{Discussion}

In this paper, we consider the problem of fidelity estimation when the reference quantum state is known. 
As a result, two quantum algorithms are proposed: one is with sample complexity $O\rbra{r^2/\varepsilon^2}$ when the reference state is of rank $r$, achieving optimal dependence on $\varepsilon$, and the other is with sample complexity $O\rbra{r^2/\varepsilon^4}$ when the unknown state is guaranteed to be of rank $r$. 
Both quantum algorithms improve the previous best results due to \cite{UNWT25} under their respective scenarios. 
As an application, we present an approach to tolerant quantum state certification, generalizing the exact quantum state certification studied in \cite{BOW19}. 
We leave some questions for future research. 

\begin{itemize}
    \item Can we close the gap between the upper bound $O\rbra{r^2/\varepsilon^2}$ and lower bound $\Omega\rbra{r/\varepsilon^2}$ in \cref{thm:fidelity-main} for fidelity estimation with a low-rank reference state? 
    Can we improve the sample complexity $O\rbra{r^2/\varepsilon^4}$ in \cref{thm:fidelity-ext} for fidelity estimation for low-rank unknown states?
    
    \item In addition to fidelity, trace distance \cite{Hel67,Hel69,Rus94} is another measure of closeness commonly used in quantum physics. The problem of quantum state certification in terms of the trace distance was also considered in \cite{BOW19}.
    Later, trace distance estimation was studied in \cite{WGL+24,WZ24b}, with the tolerant quantum state certification with respect to the trace distance as an application.
    Can we find better approaches to trace distance estimation (and thus tolerant quantum state certification) when the reference state is known?
\end{itemize}

\subsection{Independent and simultaneous work}

The independent and simultaneous work of Lowe and Tan \cite{LT26} presented a quantum algorithm that estimates the fidelity $\mathrm{F}\rbra{\rho, \sigma}$ between two unknown quantum states $\rho$ and $\sigma$ with sample complexity $O\rbra{r^2/\varepsilon^2}$, where $r$ is the maximum rank of both $\rho$ and $\sigma$.
In contrast, our results with known reference state $\sigma$ in \cref{thm:fidelity-main,thm:fidelity-ext} require prior knowledge only of the rank of either $\rho$ or $\sigma$. 

\section{Preliminaries}

\paragraph{Notations.}
Let $\ket{\psi}$ denote a complex-valued vector and $\bra{\psi}$ be its Hermitian conjugate. 
Let $\braket{\varphi}{\psi} = \bra{\varphi} \cdot \ket{\psi}$ denote the inner product of $\ket{\varphi}$ and $\ket{\psi}$. 
The trace norm of a matrix $A$ is defined as 
$\|A\|_1=\tr(\sqrt{A^\dag A})$,
i.e., the sum of singular values of $A$, where $A^\dag$ is the Hermitian conjugate of $A$.
The operator norm of a matrix $A$ is defined as 
$\Abs{A}_{\infty} = \max_{\braket{\psi}{\psi} = 1} \sqrt{\bra{\psi}A^\dag A\ket{\psi}}$, i.e., the largest singular value of $A$. 
For an Hermitian matrix $A$ with the spectral decomposition $A=\sum_{i}\lambda_i\ketbra{\psi_i}{\psi_i}$ and $\lambda_i\neq 0$, we use $\supp(A)$ to denote its support space: $\spanspace(\{\ket{\psi_i}\}_i)$. 

\subsection{Closeness measures of quantum states}

\begin{definition}[Fidelity]
The fidelity of quantum states $\rho$ and $\sigma$ is defined as
\[\mathrm{F}(\rho,\sigma)=\|\sqrt{\rho}\sqrt{\sigma}\|_1=\tr\!\left(\sqrt{\sqrt{\sigma}\rho\sqrt{\sigma}}\right)\]
\end{definition}

\begin{remark}
We also extend the definition of fidelity to partial density operators $\rho,\sigma$ (i.e., $\rho,\sigma$ are positive semidefinite operators with $\tr(\rho)\leq 1$ and $\tr(\sigma)\leq 1$).
\end{remark}

Suppose $\rho,\sigma$ are non-zero partial density operators. We can easily see that the fidelity $\mathrm{F}(\rho,\sigma)$ has the following properties.
\begin{itemize}
\item It can be calculated locally: let $\Pi$ be the orthogonal projector onto the subspace $\supp(\sigma)$, then
\[\mathrm{F}(\rho,\sigma)=\tr\!\left(\sqrt{\sqrt{\sigma}\Pi\rho\Pi\sqrt{\sigma}}\right)=\mathrm{F}(\Pi\rho\Pi,\sigma).\]
\item It is upper bounded by the square roots of traces:
\[\mathrm{F}(\rho,\sigma)= \sqrt{\tr(\rho)}\sqrt{\tr(\sigma)}\cdot \mathrm{F}\!\left(\frac{\rho}{\tr(\rho)},\frac{\sigma}{\tr(\sigma)}\right)\leq \sqrt{\tr(\rho)}\sqrt{\tr(\sigma)}.\]
\end{itemize}

\begin{definition}[Bures distance \cite{Bur69}]
The Bures distance between quantum states $\rho$ and $\sigma$ is defined as
\[\mathrm{D}_\mathrm{B}(\rho,\sigma)=\sqrt{2(1-\mathrm{F}(\rho,\sigma))}.\]
\end{definition}

The Bures distance is a valid distance measure since it satisfies the triangle inequality: for any quantum states $\rho,\sigma$ and $\eta$, we have
\[\mathrm{D}_\mathrm{B}(\rho,\sigma)+\mathrm{D}_\mathrm{B}(\sigma,\eta)\geq \mathrm{D}_\mathrm{B}(\rho,\eta).\]
Note that we define the Bures distance only for density operators (i.e., $\tr(\rho)=1$), not for proper partial density operators (i.e., $\tr(\rho)<1$).

\subsection{Quantum state tomography in Bures distance}

We need the recent result \cite{PSW25} for quantum state tomography in Bures distance, which improves the previous approach in \cite{HHJ+17} by removing the logarithmic factors in the sample complexity. This result also implies the optimal sample complexity of trace distance tomography given in \cite{OW16}. 

\begin{theorem}[{\cite[Theorem 4.1]{PSW25}}]\label{thm-12160130}
There exists a universal constant $C>0$ such that for any $\varepsilon\in (0,1)$ and $n\geq C rd/\varepsilon^2$, there exists an algorithm that uses $n$ copies of an unknown state $\rho$ with rank at most $r$ and outputs an estimate $\widehat{\rho}$ such that $\mathrm{D}_\mathrm{B}(\rho,\widehat{\rho})\leq \varepsilon$ with probability at least $99/100$.
\end{theorem}

For our purposes, we restate \cref{thm-12160130} for the case where the state is known to lie in a fixed subspace. In other words, if we have prior information that the unknown state is supported on a particular subspace, then the sample complexity required for tomography can be improved.

\begin{corollary}[Bures distance tomography restricted to a known subspace]\label{coro-1151816}
Suppose $\rho$ is an unknown $d$-dimensional rank-$r$ state with support contained in a known $d'$-dimensional subspace. Then there exists a universal constant $C>0$ such that for any $\varepsilon \in (0,1)$ and $n\geq Crd'/\varepsilon^2$, there exists an algorithm that uses $n$ copies of an unknown state $\rho$ and outputs an estimate $\widehat{\rho}$ such that $\mathrm{D}_\mathrm{B}(\rho,\widehat{\rho})\leq \varepsilon$ with probability at least $99/100$.
\end{corollary}

\section{Upper Bounds}

We state our upper bound results for fidelity estimation in the following. 

\begin{theorem}\label{thm-11242222}
Suppose $\sigma$ is a known quantum state of rank $r$, and $\rho$ is an unknown quantum state. Then, $\mathrm{F}(\rho,\sigma)$ can be estimated to within error $\varepsilon$ using $O(r^2/\varepsilon^2)$ copies of $\rho$.
\end{theorem}

\begin{theorem}\label{thm-1131746}
Suppose $\sigma$ is a known quantum state, and $\rho$ is an unknown quantum state of rank at most $r$. Then, $\mathrm{F}(\rho,\sigma)$ can be estimated to within error $\varepsilon$ using $O(r^2/\varepsilon^4)$ copies of $\rho$.
\end{theorem}

\subsection{Proof of Theorem \ref{thm-11242222}}

The algorithm is given in \cref{alg-11250952}.
\begin{algorithm}[H]
\caption{\raggedright Estimating $\mathrm{F}(\rho,\sigma)$ with low rank reference state $\sigma$}\label{alg-11250952}
    \begin{algorithmic}[1]
    \Require $n$ samples of $\rho$, classical description of rank-$r$ reference state $\sigma$.
    \Ensure An estimate of $\mathrm{F}(\rho,\sigma)$.
    \State Let $\Pi$ be the orthogonal projector onto the $r$-dimensional subspace $\supp(\sigma)$.
    \State Perform measurement $\{M_0=\Pi,M_1=I-\Pi\}$ on each of the $n$ copies of $\rho$, obtaining a sequence of measurement outcome $\{b_1,\ldots,b_n\}$ where $b_i\in\{0,1\}$.
    \State Collect all, say $m$, post-measurement states corresponding to the outcome $0$, i.e., $\left(\frac{\Pi \rho\Pi}{\tr(\Pi \rho)}\right)^{\otimes m}$.
    \State Perform the full-rank Bures distance tomography restricted to the $r$-dimensional subspace $\supp(\sigma)$ on these $m$ post-measurement states to obtain an estimate $\widehat{\rho}$.
    \State \textbf{Return} $\sqrt{\frac{m}{n}}\mathrm{F}(\widehat{\rho},\sigma)$.
    \end{algorithmic}
\end{algorithm}

Suppose $C$ is the universal constants in \cref{coro-1151816}, and without loss of generality, we assume $C\geq 10000$. Let $n=Cr^2/\varepsilon^2$, and let $\widehat{\mathrm{F}}$ denote the output of \cref{alg-11250952}. We claim that $|\widehat{\mathrm{F}}-\mathrm{F}(\rho,\sigma)|\leq O(\varepsilon)$.

Let $m \coloneqq |\{i\,|\, b_i=0\}|$. 
Note that $m$ is a random variable subject to binomial distribution with mean $n\tr(\Pi\rho)$. Thus we know that the variance $\Var(m)=n\tr(\Pi\rho)(1-\tr(\Pi\rho))\leq n\tr(\Pi\rho)$. By Chebyshev's inequality, we have
$
\Pr\left[\left|m-n\tr(\Pi\rho)\right|\geq 100 \sqrt{n\tr(\Pi\rho)}\right]\leq \frac{1}{100}.
$
This means with probability at least $99/100$, we have 
\begin{align}
\left|\frac{m}{n}-\tr(\Pi\rho)\right|\leq 100 \sqrt{\frac{\tr(\Pi\rho)}{n}}.\label{eq-12220105}
\end{align}
Now, we consider the cases where the event in \cref{eq-12220105} happens.
Then, we further have
\begin{equation}\label{eq-1201612}
\begin{split}
\left|\sqrt{\frac{m}{n}}-\sqrt{\tr(\Pi\rho)}\right|\sqrt{\tr(\Pi\rho)} & \leq \left|\sqrt{\frac{m}{n}}-\sqrt{\tr(\Pi\rho)}\right|\left|\sqrt{\frac{m}{n}}+\sqrt{\tr(\Pi\rho)}\right|  \\
&= \left|\frac{m}{n}-\tr(\Pi\rho)\right|  \\
&\leq 100\sqrt{\frac{\tr(\Pi\rho)}{n}}, 
\end{split}
\end{equation}
and thus
\begin{equation}\label{eq-12231600}
\left|\sqrt{\frac{m}{n}}-\sqrt{\tr(\Pi\rho)}\right|\leq 100\frac{1}{\sqrt{n}}\leq \frac{100\varepsilon}{\sqrt{C}r}\leq \varepsilon.
\end{equation}
Let
\[\rho_\Pi\coloneqq \frac{\Pi\rho\Pi}{\tr(\Pi\rho)},\]
and we collect all the $m$ post-measurement states corresponding to outcome $0$, which is $\rho_\Pi^{\otimes m}$.

If $\sqrt{\tr(\Pi\rho)}\leq 2\varepsilon$, then $\mathrm{F}(\rho,\sigma)=\mathrm{F}(\Pi\rho\Pi,\sigma)\leq \sqrt{\tr(\Pi\rho\Pi)}\leq 2\varepsilon$. 
We also have $\widehat{\mathrm{F}}=\sqrt{\frac{m}{n}}\mathrm{F}(\widehat{\rho},\sigma)\leq \sqrt{\frac{m}{n}}\leq \sqrt{\tr(\Pi\rho)}+\varepsilon\leq 3\varepsilon$.
Therefore, $|\widehat{\mathrm{F}}-\mathrm{F}(\rho,\sigma)|\leq 3\varepsilon$.

Otherwise, we assume 
\begin{equation}\label{eq-1131742}
\sqrt{\tr(\Pi\rho)}>2\varepsilon.
\end{equation} 
We perform the full-rank Bures distance quantum state tomography restricted to the $r$-dimensional subspace $\supp(\sigma)$ on $m$ copies of $\rho_\Pi$. 
Then, we can obtain an estimate $\widehat{\rho}$ of $\rho_\Pi$ with Bures distance error at most
\[\mathrm{D}_\mathrm{B}(\widehat{\rho},\rho_\Pi)\leq \frac{\sqrt{C}r}{\sqrt{m}}\leq \frac{\sqrt{C}r}{\left(\sqrt{\tr(\Pi\rho)}-\varepsilon\right)\sqrt{n}}=  \frac{\varepsilon}{\sqrt{\tr(\Pi\rho)}-\varepsilon}=:\varepsilon',\]
where the first inequality is by \cref{coro-1151816}, the second inequality is by \cref{eq-12231600}, and the first equality is by the definition of $n$.
This means
\[|\mathrm{D}_\mathrm{B}(\widehat{\rho},\sigma)-\mathrm{D}_\mathrm{B}(\rho_\Pi,\sigma)|\leq \mathrm{D}_\mathrm{B}(\widehat{\rho},\rho_\Pi)\leq \varepsilon'.\]
Thus,
\begin{align}
|\mathrm{F}(\widehat{\rho},\sigma)-\mathrm{F}(\rho_\Pi,\sigma)|&=\frac{1}{2}\left|\mathrm{D}_\mathrm{B}(\widehat{\rho},\sigma)^2-\mathrm{D}_\mathrm{B}(\rho_\Pi,\sigma)^2\right|\nonumber \\
&\leq \sqrt{2}\varepsilon',\label{eq-1131736}
\end{align}
where \cref{eq-1131736} uses the fact that the Bures distance is at most $\sqrt{2}$.
Therefore, we have
\begin{align}
\left|\widehat{\mathrm{F}}-\mathrm{F}(\rho,\sigma)\right|&=\left|\sqrt{\frac{m}{n}}\mathrm{F}(\widehat{\rho},\sigma)- \mathrm{F}(\Pi\rho\Pi,\sigma)\right| \nonumber\\
&=\left|\sqrt{\frac{m}{n}}\mathrm{F}(\widehat{\rho},\sigma)- \sqrt{\tr(\Pi\rho)}\mathrm{F}(\rho_\Pi,\sigma)\right| \nonumber \\
&\leq \left|\sqrt{\frac{m}{n}}- \sqrt{\tr(\Pi\rho)}\right|\mathrm{F}(\widehat{\rho},\sigma)+ \sqrt{\tr(\Pi\rho)}\left|\mathrm{F}(\widehat{\rho},\sigma)- \mathrm{F}(\rho_\Pi,\sigma)\right| \nonumber \\
&\leq \varepsilon+\sqrt{2\tr(\Pi\rho)}\varepsilon' \label{eq-1131739}\\
&= \varepsilon +\frac{\sqrt{2\tr(\Pi\rho)}\varepsilon}{\sqrt{\tr(\Pi\rho)}-\varepsilon} \nonumber \\
&\leq \varepsilon + \frac{2\sqrt{2}\varepsilon^2}{\varepsilon} \label{eq-1131743}\\
&=O(\varepsilon),\nonumber 
\end{align}
where \cref{eq-1131739} is due to \cref{eq-12231600} and \cref{eq-1131736}, \cref{eq-1131743} is due to \cref{eq-1131742}.

\subsection{Proof of Theorem \ref{thm-1131746}}
The algorithm is given in \cref{alg-1131806}.
\begin{algorithm}[H]
\caption{\raggedright Estimating $\mathrm{F}(\rho,\sigma)$ with low rank unknown state $\rho$}\label{alg-1131806}
    \begin{algorithmic}[1]
    \Require $n$ samples of $\rho$ with rank at most $r$, classical description of the reference state $\sigma$.
    \Ensure An estimate of $\mathrm{F}(\rho,\sigma)$.
    \State Set $K=r/\varepsilon^2$.
    \State Let $\Pi$ be the orthogonal projector onto the $K$-dimensional eigenspace of the top-$K$ eigenvalues of $\sigma$.
    \State Perform measurement $\{M_0=\Pi,M_1=I-\Pi\}$ on each of the $n$ copies of $\rho$, obtaining a sequence of measurement outcome $\{b_1,\ldots,b_n\}$ where $b_i\in\{0,1\}$.
    \State Collect all, say $m$, post-measurement states corresponding to the outcome $0$, i.e., $\left(\frac{\Pi \rho\Pi}{\tr(\Pi \rho)}\right)^{\otimes m}$.
    \State Perform the rank-$r$ Bures distance tomography restricted to the $K$-dimensional subspace $\supp(\Pi)$ on these $m$ post-measurement states to obtain an estimate $\widehat{\rho}$.
    \State \textbf{Return} $\sqrt{\frac{m}{n}}\mathrm{F}(\widehat{\rho},\sigma)$.
    \end{algorithmic}
\end{algorithm}

Suppose $C$ is the universal constants in \cref{coro-1151816}, and without loss of generality, we assume $C\geq 10000$. Let $n=C rK/\varepsilon^2=Cr^2/\varepsilon^4$, and let $\widehat{\mathrm{F}}$ denote the output of \cref{alg-1131806}. We claim that $|\widehat{\mathrm{F}}-\mathrm{F}(\rho,\sigma)|\leq O(\varepsilon)$.

Let $m \coloneqq |\{i\,|\, b_i=0\}|$. 
Note that $m$ is a random variable subject to binomial distribution with mean $n\tr(\Pi\rho)$. Thus we know that the variance $\Var(m)=n\tr(\Pi\rho)(1-\tr(\Pi\rho))\leq n\tr(\Pi\rho)$. By Chebyshev's inequality, we have
$
\Pr\left[\left|m-n\tr(\Pi\rho)\right|\geq 100 \sqrt{n\tr(\Pi\rho)}\right]\leq \frac{1}{100}.
$
This means with probability at least $99/100$, we have 
\begin{align}
\left|\frac{m}{n}-\tr(\Pi\rho)\right|\leq 100 \sqrt{\frac{\tr(\Pi\rho)}{n}}.\label{eq-1201606}
\end{align}
Now, we consider the cases where the event in \cref{eq-1201606} happens.
Similar to that the argument used in \cref{eq-1201612}, we have
\begin{equation}\label{eq-1201608}
\left|\sqrt{\frac{m}{n}}-\sqrt{\tr(\Pi \rho)}\right|\leq 100\frac{1}{\sqrt{n}}\leq  \varepsilon^2\leq \varepsilon.
\end{equation}
Let
\[\rho_\Pi\coloneqq \frac{\Pi\rho\Pi}{\tr(\Pi\rho)},\]
and we collect all the $m$ post-measurement states corresponding to outcome $0$, which is $\rho_\Pi^{\otimes m}$. 

Let $\lambda_1\geq \cdots \geq \lambda_d$ be the eigenvalues of $\sigma$ and $\ket{\psi_i}$ are the corresponding eigenvectors.
We define $\Lambda=\sum_{i=1}^k \lambda_i$ and
\[
\sigma'=\Pi\sigma\Pi =\sum_{i=1}^{K} \lambda_i \ketbra{\psi_i}{\psi_i}.
\]
Note that $\sigma'$ is not neccesarily a density oeprator but only a partial density operator (i.e., $\tr(\sigma')\leq 1$).

If $\sqrt{\tr(\Pi\rho)}\leq 2\varepsilon$, then $\mathrm{F}(\rho,\sigma')=\mathrm{F}(\Pi\rho\Pi,\sigma')\leq \sqrt{\tr(\Pi\rho\Pi)}\sqrt{\tr(\sigma')}\leq 2\varepsilon$.
Thus by \cref{lemma-1210416}, we know $F(\rho,\sigma)\leq 3\varepsilon$.
We also have $\widehat{\mathrm{F}}=\sqrt{\frac{m}{n}}\mathrm{F}(\widehat{\rho},\sigma)\leq \sqrt{\frac{m}{n}} \leq \sqrt{\tr(\Pi\rho)}+\varepsilon\leq 3\varepsilon$.
Therefore, we have $|\widehat{\mathrm{F}}-\mathrm{F}(\rho,\sigma)|\leq 3\varepsilon$.

Otherwise, we assume 
\begin{equation}\label{eq-1201630}
\sqrt{\tr(\Pi\rho)}>2\varepsilon.
\end{equation} 
We perform the rank-$r$ Bures distance quantum state tomography restricted to the $K$-dimensional subspace $\supp(\Pi)$ on $m$ copies of $\rho_\Pi$. 
Then, with probability at least $99/100$, we can obtain an estimate $\widehat{\rho}$ of $\rho_\Pi$ with Bures distance error at most
\[\mathrm{D}_\mathrm{B}(\widehat{\rho},\rho_\Pi)\leq \frac{\sqrt{CrK}}{\sqrt{m}}\leq \frac{\sqrt{CrK}}{\left(\sqrt{\tr(\Pi\rho)}-\varepsilon\right)\sqrt{n}}=  \frac{\varepsilon}{\sqrt{\tr(\Pi\rho)}-\varepsilon}=:\varepsilon',\]
where the first inequality is by \cref{coro-1151816}, the second inequality is by \cref{eq-1201608}, and the first equality is by the definition of $n$.
This means
\[|\mathrm{D}_\mathrm{B}(\widehat{\rho},\sigma)-\mathrm{D}_\mathrm{B}(\rho_\Pi,\sigma)|\leq \mathrm{D}_\mathrm{B}(\widehat{\rho},\rho_\Pi)\leq \varepsilon'.\]
Thus,
\begin{align}
|\mathrm{F}(\widehat{\rho},\sigma)-\mathrm{F}(\rho_\Pi,\sigma)|&=\frac{1}{2}\left|\mathrm{D}_\mathrm{B}(\widehat{\rho},\sigma)^2-\mathrm{D}_\mathrm{B}(\rho_\Pi,\sigma)^2\right|\nonumber \\
&\leq \sqrt{2}\varepsilon',\label{eq-1210433}
\end{align}
where \cref{eq-1210433} uses the fact that the Bures distance is at most $\sqrt{2}$.
Therefore, we have
\begin{align}
\left|\widehat{\mathrm{F}}-\mathrm{F}(\rho,\sigma')\right|&=\left|\sqrt{\frac{m}{n}}\mathrm{F}(\widehat{\rho},\sigma)- \mathrm{F}(\Pi\rho\Pi,\Pi\sigma\Pi)\right| \nonumber\\
&=\left|\sqrt{\frac{m}{n}}\mathrm{F}(\widehat{\rho},\sigma)- \mathrm{F}(\Pi\rho\Pi,\sigma)\right| \nonumber\\
&=\left|\sqrt{\frac{m}{n}}\mathrm{F}(\widehat{\rho},\sigma)- \sqrt{\tr(\Pi\rho)}\mathrm{F}(\rho_\Pi,\sigma)\right| \nonumber \\
&\leq \left|\sqrt{\frac{m}{n}}- \sqrt{\tr(\Pi\rho)}\right|\mathrm{F}(\widehat{\rho},\sigma)+ \sqrt{\tr(\Pi\rho)}\left|\mathrm{F}(\widehat{\rho},\sigma)- \mathrm{F}(\rho_\Pi,\sigma)\right| \nonumber \\
&\leq \varepsilon+\sqrt{2\tr(\Pi\rho)}\varepsilon' \label{eq-212343} \\
&= \varepsilon +\frac{\sqrt{2\tr(\Pi\rho)}\varepsilon}{\sqrt{\tr(\Pi\rho)}-\varepsilon} \nonumber \\
&\leq \varepsilon + \frac{2\sqrt{2}\varepsilon^2}{\varepsilon} \label{eq-212345}\\
&=O(\varepsilon),\nonumber 
\end{align}
where \cref{eq-212343} is due to \cref{eq-1201608} and \cref{eq-1210433}, and \cref{eq-212345} is due to \cref{eq-1201630}.
Then, using \cref{lemma-1210416}, we can conclude
\[\left|\widehat{\mathrm{F}}-\mathrm{F}(\rho,\sigma)\right|\leq \left|\widehat{\mathrm{F}}-\mathrm{F}(\rho,\sigma')\right|+\left|\mathrm{F}(\rho,\sigma')-\mathrm{F}(\rho,\sigma)\right|\leq O(\varepsilon).\]

\begin{lemma}\label{lemma-1210416}
Let $K=\ceil{r/\varepsilon^2}$. Let $\sigma,\rho$ be $d$-dimensional states, where $\rho$ has rank at most $r$. 
Let $\lambda_1\geq \cdots\geq \lambda_d$ be the eigenvalues of $\sigma$ and $\ket{\psi_i}$ are the corresponding eigenvectors. Define $\sigma'=\sum_{i=1}^K \lambda_i\ketbra{\psi_i}{\psi_i}$, then we have
\[|\mathrm{F}(\rho,\sigma')-\mathrm{F}(\rho,\sigma)|\leq \varepsilon.\]
\end{lemma}
\begin{proof}
Let $\sigma''=\sigma-\sigma'$. Note that for any $i\geq K+1$, we have $\lambda_{i}\leq 1/K\leq \varepsilon^2/r$. Therefore, $\|\sigma''\|_\infty\leq \varepsilon^2/r$. This means
\begin{align}
|\mathrm{F}(\rho,\sigma')-\mathrm{F}(\rho,\sigma)|&=\left|\left\|\sqrt{\rho}\sqrt{\sigma'}\right\|_1-\left\|\sqrt{\rho}\sqrt{\sigma}\right\|_1\right| \nonumber \\
&=\left|\left\|\sqrt{\rho}\sqrt{\sigma'}\right\|_1-\left\|\sqrt{\rho}\sqrt{\sigma'}+\sqrt{\rho}\sqrt{\sigma''}\right\|_1\right|\nonumber \\
&\leq \|\sqrt{\rho}\sqrt{\sigma''}\|_1\nonumber \\
&\leq (\varepsilon/\sqrt{r})\cdot  \|\sqrt{\rho}\|_1 \label{eq-1210412} \\
&\leq \varepsilon, \label{eq-1210413}
\end{align}
where \cref{eq-1210412} is by the matrix H\"older inequality $\|AB\|_1\leq \|A\|_\infty\|B\|_1$ and \cref{eq-1210413} is because $\rho$ is of rank at most $r$.
\end{proof}

\section{Lower Bounds}

\begin{theorem}\label{thm-240417}
    For any $\varepsilon \in \rbra{0, \frac{1}{10}}$ and any integer $r \geq 1$, any quantum algorithm that estimates the fidelity $\mathrm{F}\rbra{\rho, \sigma}$ to within additive error $\varepsilon$ for a known quantum state $\sigma$ of rank $r$ and an unknown quantum state $\rho$ requires $\Omega\rbra{\frac{r}{\varepsilon^2}}$ samples of $\rho$. 
\end{theorem}

\begin{proof}
    \textbf{The hard problem for reduction.}
    We consider the following problem $\textsc{UnifDis}_r$ for an integer $r \geq 1$, which was initially studied in \cite{CHW07}. Given an unknown $2r$-dimensional quantum state $\rho$, determine whether:
    \begin{enumerate}
        \item $\rho$ has spectrum $\rbra{\underbrace{\tfrac{1}{r}, \dots, \tfrac{1}{r}}_{r}, \underbrace{0, \dots, 0}_{r}}$, or \label{item:case1}
        \item $\rho$ is the maximally mixed state $\frac{I}{2r}$, i.e., with spectrum $\rbra{\underbrace{\tfrac{1}{2r}, \dots, \tfrac{1}{2r}}_{2r}}$. \label{item:case2}
    \end{enumerate}
    According to \cite[Theorem 3]{CHW07} (see also \cite[Theorem 1.9]{OW21}), any quantum algorithm that solves the problem $\textsc{UnifDis}_r$ requires $\Omega\rbra{r}$ samples of $\rho$ for any integer $r \geq 1$. 
    Our proof is inspired by the quantum sample lower bound for rank testing given in \cite[Section 6.2]{OW21}, where they also reduced the problem of rank testing to the problem $\textsc{UnifDis}_r$. 

    \textbf{Formalism of fidelity estimation.} We will establish a sample lower bound for fidelity estimation by reducing the problem $\textsc{UnifDis}_r$. 
    To this end, suppose there is a quantum algorithm $\mathcal{A}$ that given a known quantum state $\sigma$ of rank $r$ and an unknown quantum state $\rho$, computes the fidelity $\mathrm{F}\rbra{\rho, \sigma}$ to within additive error $\varepsilon$ using $\mathsf{S}\rbra{r, \varepsilon}$ samples of $\rho$. 
    Formally, let $\mathcal{A}_\sigma\rbra{\rho^{\otimes \mathsf{S}\rbra{r, \varepsilon}}}$ denote the output real random number of the algorithm $\mathcal{A}$ on input $\mathsf{S}\rbra{r, \varepsilon}$ identical copies of $\rho$, which satisfies
    \begin{equation} \label{eq:def-fid}
        \Pr\sbra*{ \abs*{\mathcal{A}_\sigma\rbra{\rho^{\otimes \mathsf{S}\rbra{r, \varepsilon}}} - \mathrm{F}\rbra{\rho, \sigma}} \leq \varepsilon } \geq 0.99.
    \end{equation}
    Here, $\mathcal{A}_\sigma\rbra{\cdot}$ can be described as a quantum channel $\mathcal{E}$ followed by a computational basis measurement $\cbra{M_{m} = \ketbra{m}{m} \otimes I}$ such that on input quantum state $\varrho$, $\mathcal{A}_\sigma\rbra{\cdot}$ outputs a random real number $m$ with probability
    \begin{equation} \label{eq:def-algo-A}
        \Pr\sbra{\mathcal{A}_\sigma\rbra{\varrho} = m} = \tr\rbra*{M_m \mathcal{E}\rbra{\varrho} M_m^\dag}. 
    \end{equation}
    
    \textbf{A protocol for $\textsc{UnifDis}_r$ based on fidelity estimation.} For our purpose, we consider the fidelity estimation for $\rbra{2r+1}$-dimensional quantum states, with the computational basis $\cbra{\ket{0}, \ket{1}, \dots, \ket{2r}}$. First, set $\sigma$ to the following $\rbra{2r+1}$-dimensional quantum state:
    \[
    \sigma = \sum_{i=1}^{2r} \frac{1}{2r}\ketbra{i}{i}.
    \]
    For an unknown $2r$-dimensional quantum state $\rho$, we can determine whether it is in Case \ref{item:case1} or in Case \ref{item:case2} with the following protocol $\mathcal{P}_{r,\varepsilon}$ with $\varepsilon \in \rbra{0, 1}$ an arbitrary parameter in \cref{algo:protocol-1}. 

    \begin{algorithm}[H]
\caption{Protocol $\mathcal{P}_{r, \varepsilon}$ with parameters $\varepsilon \in \rbra{0, 1}$ and rank $r \geq 1$}\label{algo:protocol-1}
    \begin{algorithmic}[1]
    \Require Samples of an unknown quantum state $\rho$. 
    \Ensure Whether $\rho$ is in Case \ref{item:case1} or in Case \ref{item:case2}.
    \State Let $S \gets \mathsf{S}\rbra{r, \frac{\varepsilon}{10}}$.
    \State Let $x \in \cbra{0, 1}^S$ be a random binary string with each $x_i$ independent and $\Pr\sbra{x_i = 0} = 1-\varepsilon^2$. 
    \State Compute $\hat{F} \gets \mathcal{A}_\sigma\rbra{\varrho\rbra{x}}$, where 
        \[
        \varrho\rbra{x} = \bigotimes_{i=1}^S \varrho_i,
        \]
        with $\varrho_i = \ketbra{0}{0}$ if $x_i = 0$ and $\varrho_i = \rho$ if $x_i = 1$. 
    \State \textbf{Return} Case \ref{item:case1} if $\hat{F}_x < 0.85 \varepsilon$ and Case \ref{item:case2} otherwise. 
    \end{algorithmic}
\end{algorithm}

    \textbf{Correctness of the protocol $\mathcal{P}_{r, \varepsilon}$.}
    To show the correctness of the protocol $\mathcal{P}_{r, \varepsilon}$, we first note that the probability distribution of $\mathcal{A}_\sigma\rbra{\E_x\sbra{\varrho\rbra{x}}}$ is equal to the probability distribution of $\mathcal{A}_\sigma\rbra{\varrho\rbra{x}}$ when $x$ is uniformly random, i.e., for any possible output $m$, 
    \begin{equation} \label{eq:prob-eq}
        \Pr\sbra*{\mathcal{A}_\sigma\rbra*{\E_x\sbra{\varrho\rbra{x}}} = m} = \Pr_x \sbra*{ \mathcal{A}_\sigma\rbra{\varrho\rbra{x}} = m }. 
    \end{equation}
    To see this, by \cref{eq:def-algo-A}, we have 
    \begin{align*}
        \Pr\sbra*{\mathcal{A}_\sigma\rbra*{\E_x\sbra{\varrho\rbra{x}}} = m}
        & = \tr\rbra*{M_m \mathcal{E}\rbra*{\E_x\sbra{\varrho\rbra{x}}} M_m^\dag} \\
        & = \E_x\sbra*{ \tr\rbra*{M_m \mathcal{E}\rbra*{\varrho\rbra{x}} M_m^\dag} } \\
        & = \E_x\sbra*{\Pr\sbra{\mathcal{A}_\sigma\rbra{\varrho\rbra{x}} = m}} \\
        & = \Pr_x \sbra*{\mathcal{A}_\sigma\rbra{\varrho\rbra{x}} = m}.
    \end{align*}
    Let $\rho_\varepsilon = \rbra{1-\varepsilon^2} \ketbra{0}{0} + \varepsilon^2 \rho = \E_x \sbra{\varrho_i}$ for each $1 \leq i \leq S$. 
    Then,
    \begin{equation} \label{eq:prob-hat-F}
        \Pr_x \sbra*{\abs*{\mathcal{A}_\sigma\rbra{\varrho\rbra{x}}-\mathrm{F}\rbra{\rho_\varepsilon, \sigma}} \leq \frac{\varepsilon}{10}} \geq 0.99.
    \end{equation}
    To see this, note that $\E_x \sbra{\varrho\rbra{x}} = \rho_\varepsilon^{\otimes S}$, and we further have
    \begin{align}
        \Pr_x \sbra*{\abs*{\mathcal{A}_\sigma\rbra{\varrho\rbra{x}}-\mathrm{F}\rbra{\rho_\varepsilon, \sigma}} \leq \frac{\varepsilon}{10}} 
        & = \sum_{m \colon \abs{m - \mathrm{F}\rbra{\rho_\varepsilon, \sigma}} \leq \frac{\varepsilon}{10}} \Pr_x \sbra{\mathcal{A}_\sigma\rbra{\varrho\rbra{x}}=m} \nonumber \\
        & = \sum_{m \colon \abs{m - \mathrm{F}\rbra{\rho_\varepsilon, \sigma}} \leq \frac{\varepsilon}{10}} \Pr\sbra*{\mathcal{A}_\sigma\rbra*{\E_x\sbra{\varrho\rbra{x}}} = m} \label{eq:use-def-prob} \\
        & = \Pr\sbra*{\abs*{\mathcal{A}_\sigma\rbra*{\E_x\sbra{\varrho\rbra{x}}}-\mathrm{F}\rbra{\rho_\varepsilon, \sigma}} \leq \frac{\varepsilon}{10}} \nonumber \\
        & = \Pr\sbra*{\abs*{\mathcal{A}_\sigma\rbra*{\rho_\varepsilon^{\otimes S}}-\mathrm{F}\rbra{\rho_\varepsilon, \sigma}} \leq \frac{\varepsilon}{10}} \nonumber \\
        & \geq 0.99, \label{eq:by-def-fid}
    \end{align}
    where \cref{eq:use-def-prob} uses \cref{eq:prob-eq} and \cref{eq:by-def-fid} is by \cref{eq:def-fid}. 

    Now we are ready to show the correctness of the protocol $\mathcal{P}_{r, \varepsilon}$. 
    \begin{enumerate}
        \item If $\rho$ is in Case \ref{item:case1}, then \begin{align*}
            \mathrm{F}\rbra{\rho_\varepsilon, \sigma} = r \cdot \sqrt{\frac{\varepsilon^2}{r} \cdot \frac{1}{2r}} = \frac{\varepsilon}{\sqrt{2}}. 
        \end{align*}
        By \cref{eq:prob-hat-F}, 
        \[
        \Pr_x \sbra*{\abs*{\mathcal{A}_\sigma\rbra{\varrho\rbra{x}} - \frac{\varepsilon}{\sqrt{2}}} \leq \frac{\varepsilon}{10}} \geq 0.99,
        \]
        which gives
        \begin{equation} \label{eq:prob-case-1}
            \Pr_x \sbra* {\mathcal{A}_\sigma\rbra{\varrho\rbra{x}} < 0.85\varepsilon} \geq 0.99.
        \end{equation}
        
        \item If $\rho$ is in Case \ref{item:case2}, then \begin{align*}
            \mathrm{F}\rbra{\rho_\varepsilon, \sigma} = 2r \cdot \sqrt{\frac{\varepsilon^2}{2r} \cdot \frac{1}{2r}} = \varepsilon. 
        \end{align*}
        By \cref{eq:prob-hat-F}, 
        \[
        \Pr_x \sbra*{\abs*{\mathcal{A}_\sigma\rbra{\varrho\rbra{x}} - \varepsilon} \leq \frac{\varepsilon}{10}} \geq 0.99,
        \]
        which gives
        \begin{equation} \label{eq:prob-case-2}
        \Pr_x \sbra*{\mathcal{A}_\sigma\rbra{\varrho\rbra{x}} > 0.85\varepsilon} \geq 0.99.
        \end{equation}
    \end{enumerate}
    Therefore, the protocol $\mathcal{P}_{r, \varepsilon}$ solves the problem $\textsc{UnifDis}_r$ with success probability at least $0.99$. 

    \textbf{Derandomization of the protocol $\mathcal{P}_{r, \varepsilon}$.} Now we look into the details of the protocol $\mathcal{P}_{r, \varepsilon}$ regarding its randomness. 
    Note that $\E_x \sbra{\Abs{x}_1} = S\varepsilon^2$, where $\Abs{x}_1 = \sum_{i=1}^S \abs{x_i}$. 
    By Markov's inequality, $\Pr\sbra{\Abs{x}_1 \leq 100 S \varepsilon^2} \geq 0.99$. 
    Therefore, there must exist an $x^* \in \cbra{0, 1}^S$ with $\Abs{x^*}_1 \leq 100S\varepsilon^2$ such that one can solve the problem $\textsc{UnifDis}_r$ with success probability at least $0.9$ using the output random number $\mathcal{A}_\sigma\rbra{\varrho\rbra{x^*}}$, i.e., both the following conditions hold:
    \begin{align*}
        \Pr\sbra{\mathcal{A}_\sigma\rbra{\varrho\rbra{x^*}} < 0.85\varepsilon \mid \rho \textup{ is in Case \ref{item:case1}}} & \geq 0.9, \\
        \Pr\sbra{\mathcal{A}_\sigma\rbra{\varrho\rbra{x^*}} > 0.85\varepsilon \mid \rho \textup{ is in Case \ref{item:case2}}} & \geq 0.9.
    \end{align*}
    Otherwise, if the first condition does not hold, then when $\rho$ is in Case \ref{item:case1}, 
    \begin{align*}
        \Pr_x \sbra{\mathcal{A}_\sigma\rbra{\varrho\rbra{x}} < 0.85\varepsilon} 
        & = \sum_{\Abs{u}_1 \leq 100 S \varepsilon^2} \Pr_x \sbra{x = u} \cdot \Pr \sbra{\mathcal{A}_\sigma\rbra{\varrho\rbra{u}} < 0.85\varepsilon} \;+ \\
        & \qquad \sum_{\Abs{u}_1 > 100 S \varepsilon^2} \Pr_x \sbra{x = u} \cdot \Pr \sbra{\mathcal{A}_\sigma\rbra{\varrho\rbra{u}} < 0.85\varepsilon} \\
        & < \Pr_x \sbra{\Abs{x}_1 \leq 100S\varepsilon^2} \cdot 0.9 + \Pr_x \sbra{\Abs{x}_1 > 100S\varepsilon^2} \cdot 1 \\
        & < 0.9 + 0.01 = 0.91,
    \end{align*}
    which contradicts \cref{eq:prob-case-1}, and, similarly, if the second condition does not hold, then when $\rho$ is in Case \ref{item:case2},
    \begin{align*}
        \Pr_x \sbra{\mathcal{A}_\sigma\rbra{\varrho\rbra{x}} > 0.85\varepsilon} < 0.91,
    \end{align*}
    which contradicts \cref{eq:prob-case-2}. 

    Now we provide a protocol $\mathcal{P}'_{r,\varepsilon}$ in \cref{algo:protocol-2} based on the previous $\mathcal{P}_{r,\varepsilon}$. 

    \begin{algorithm}[H]
\caption{Protocol $\mathcal{P}'_{r, \varepsilon}$ with parameters $\varepsilon \in \rbra{0, 1}$ and rank $r \geq 1$}\label{algo:protocol-2}
    \begin{algorithmic}[1]
    \Require Samples of an unknown quantum state $\rho$. 
    \Ensure Whether $\rho$ is in Case \ref{item:case1} or in Case \ref{item:case2}.
    \State Let $S \gets \mathsf{S}\rbra{r, \frac{\varepsilon}{10}}$.
    \State Let $x^* \in \cbra{0, 1}^S$ be the binary string defined above such that $\Abs{x^*}_1 \leq 100S\varepsilon^2$. 
    \State Compute $\hat{F} \gets \mathcal{A}_\sigma\rbra{\varrho\rbra{x^*}}$.
    \State \textbf{Return} Case \ref{item:case1} if $\hat{F} < 0.85 \varepsilon$ and Case \ref{item:case2} otherwise. 
    \end{algorithmic}
\end{algorithm}
    
    We have already shown that the protocol $\mathcal{P}'_{r,\varepsilon}$ solves the problem $\textsc{UnifDis}_r$ with success probability at least $0.9$. 
    On the other hand, the protocol $\mathcal{P}'_{r,\varepsilon}$ uses $\Abs{x^*}_1 \leq 100S\varepsilon^2$ samples of $\rho$, and thus $100 S \varepsilon^2 \geq \Omega\rbra{r}$, 
    which gives 
    \[
    S = \mathsf{S}\rbra*{r, \frac{\varepsilon}{10}} \geq \Omega\rbra*{\frac{r}{\varepsilon^2}}. 
    \]
    These complete the proof as $\varepsilon \in \rbra{0, 1}$ is arbitrary. 
\end{proof}

\section*{Acknowledgment}

The author sincerely thanks Kean Chen for discussions at an early stage of this work.

\addcontentsline{toc}{section}{References}

\bibliographystyle{alphaurl}
\bibliography{main}

\end{document}